\documentclass[letter,11pt]{article}
\usepackage[margin=1in]{geometry}
\usepackage[numbers,sort]{natbib}

\usepackage{amsmath,amsthm,amssymb,thm-restate}
\usepackage{algorithm,algpseudocode}
\usepackage{authblk}
\usepackage{hyperref}
\usepackage{cleveref}
\usepackage{xspace}
\usepackage{xcolor}
\usepackage{nicefrac}
\usepackage{thm-restate}

\allowdisplaybreaks

\newtheorem{thm}{Theorem}[section]
\newtheorem{cor}[thm]{Corollary}

\newtheorem{lem}[thm]{Lemma}

\newtheorem{cla}[thm]{Claim}

\algdef{SE}[EVENT]{Event}{EndEvent}[1]{\textbf{upon}\ #1\ \algorithmicdo}{\algorithmicend\ \textbf{event}}

\makeatletter

\algrenewcommand\algorithmiccomment[2][\normalsize]{{#1\hfill\(\triangleright\) \emph{#2}}}
\makeatother

\newcommand{\IID}{i.i.d\xspace}
\newcommand{\algoname}{\textsc{fair-bias}\xspace}
\newcommand{\E}{\mathbb{E}}
\newcommand{\STD}{\mathrm{std}}
\newcommand{\ALG}{\mathrm{ALG}}
\newcommand{\OPT}{\mathrm{OPT}}
\newcommand{\SUP}{\mathrm{SUP}}
\newcommand{\support}{\mathrm{support}}

\newcommand{\e}{\varepsilon}
\newcommand{\tsty}{\textstyle}

\newcommand{\calD}{\mathcal{D}}
\newcommand{\calU}{\mathcal{U}}
\newcommand{\calI}{\mathcal{I}}

\newcommand{\Tk}{\mathbb{T}_k}

\title{Stochastic Online Metric Matching} 

\author[1]{Anupam Gupta\thanks{Email address: \url{anupamg@cs.cmu.edu}.}\thanks{Supported in part by NSF awards CCF-1536002, CCF-1540541, and CCF-1617790, and the Indo-US Joint Center for Algorithms Under Uncertainty.}}
\author[2]{Guru Guruganesh\thanks{Email address: \url{gurug@google.com}}\thanks{Work done in part while the author was at Carnegie Mellon University.}}
\author[3]{Binghui Peng\thanks{Email address: \url{pbh15@mails.tsinghua.edu.cn}.}\thanks{Work done in part while the author was visiting Carnegie Mellon University.}}
\author[1]{David Wajc\thanks{Email address: \url{dwajc@cs.cmu.edu}.}\thanks{Supported in part by NSF grants CCF-1618280, CCF-1814603, CCF-1527110, NSF CAREER award CCF-1750808 and a Sloan Research Fellowship.}}
\affil[1]{Carnegie Mellon University and University of Pittsburgh}
\affil[2]{Google Research}
\affil[3]{Tsinghua University}

\date{\vspace{-0.5cm}}

\begin{document}

\maketitle
\begin{abstract}
  We study the minimum-cost metric perfect matching problem 
  under online i.i.d arrivals. 
  We are given a fixed metric with a server at each of the
  points, and then requests arrive online, each drawn independently from
  a known probability distribution over the points. Each request has to be
  matched to a free server, with cost equal to the distance. The goal is
  to minimize the expected total cost of the matching.

  Such stochastic arrival models have been widely studied for the
  \emph{maximization} variants of the online matching problem; however,
  the only known result for the \emph{minimization} problem is a tight
  $O(\log n)$-competitiveness for the random-order arrival model. This
  is in contrast with the adversarial model, where an optimal
  competitive ratio of $O(\log n)$ has long been conjectured and remains
  a tantalizing open question. 
	
  In this paper, we show improved results in the i.i.d arrival model.
  We show how the i.i.d model can be used to give substantially better
  algorithms: our main result is an
  $O((\log \log \log n)^2)$-competitive algorithm in this model. Along
  the way we give a $9$-competitive algorithm for the line and tree
  metrics. Both results imply a strict separation between the i.i.d model and the adversarial and random order models, both for general metrics and these much-studied metrics.
\end{abstract}

\section{Introduction}

We study the minimum-cost metric (perfect) matching problem under online
i.i.d.\ arrivals. In this problem, we are given a fixed metric $(S,d)$
with a server at each of the $n = |S|$ points. Then $n$ requests arrive
online, where each request is at a location that is drawn independently
from a known probability distribution $\calD$ over the points. Each such
arriving request has to be matched immediately and irrevocably to a free
server, whereupon it incurs a cost equal to distance of its location to
this server. The goal is to minimize the total expected cost.

The minimization version of online matching was first considered in the
standard adversarial setting by Khuller et al.~\cite{khuller1994line}
and Kalyanasundaram and Pruhs~\cite{kalyanasundaram1993online}; both
papers showed $(2n-1)$-competitive deterministic algorithms, and proved
that this was tight for, say, the star metric. After about a decade, a
randomized algorithm with an $O(\log^3 n)$-competitiveness was given by
Meyerson et al.~\cite{meyerson2006randomized}; this was improved to
$O(\log^2 n)$ by Bansal et al.~\cite{bansal2007log}, which remains the
best result known. (Recall that the maximization version of matching
problems have been very widely studied, but they use mostly unrelated
techniques.)

The competitive ratio model with adversarial online arrivals is often
considered too pessimistic, since it assumes an all-powerful adversary.
One model to level the playing field, and to make the model perhaps
closer to practice, is to restrict the adversary's power. Two models
have been popular here: the \emph{random-order arrivals} (or
\emph{secretary}) model, and the \emph{i.i.d.} model defined above. The
random-order model is a \emph{semi-random} model, in which the
worst-case input is subjected to random perturbations. Specifically, the
adversary chooses a \emph{set} of requests, which are then presented to
the algorithm in a uniformly random order. The min-cost online matching
problem in this random-order model was studied by Raghvendra, who gave a
tight $O(\log n)$-competitive algorithm~\cite{raghvendra2016robust}. The
random-order model also captures the i.i.d.\ setting, so the natural
goal is to get a better algorithm for the i.i.d.\ model.
Indeed, our main result for the i.i.d.\ model gives exactly such a result:
\begin{thm}[Main Theorem]
	\label{thm:main}
	There is an $O((\log \log \log n)^2)$-competitive algorithm for online
	minimum-cost metric perfect matching in the i.i.d.\ setting.
\end{thm}

Observe that the competitiveness here is better than the lower bounds of
$\Omega(\log n)$ known for the worst-case and random-order models. 

\medskip\noindent\textbf{Matching on the Line and Trees.}  There has also been
much interest in solving the problem for the line metric: a
deterministic lower bound of $(9+\e)$ for some $\e > 0$ is known,
showing it is strictly harder than the optimal search (or ``cow-path'')
problem, which it generalizes~\cite{fuchs2005online}. However, getting
better results for the line than for general metrics has been elusive:
an $O(\log n)$-competitive \emph{randomized} algorithm for line metrics
(and for doubling metrics) was given by~\cite{gupta2012online}. In the
\emph{deterministic} setting, recently Nayyar and
Raghvendra~\cite{nayyar2017input} gave an $O(\log^2 n)$-competitive algorithm, whose competitive ratio was subsequently proven to be $O(\log n)$ by
Raghvendra~\cite{raghvendra2018optimal},
improving on the $o(n)$-competitive algorithm of Antoniadis et
al.~\cite{antoniadis2014n}. To the best of our knowledge, nothing better
is known for tree metrics than for general metrics in both the
adversarial and the random-order models. Our second result  for the i.i.d.~model is a constant-competitive algorithm for tree metrics.
\begin{thm}[Algorithm for Trees] \label{thm:tree}
	There is a $9$-competitive algorithm
	for online minimum-cost metric perfect matching on tree metrics in the
	i.i.d.~setting.
\end{thm}

Observe that the competitiveness here is better than the lower bound of $9+\epsilon$ for line metrics in the worst-case model.

\medskip\noindent\textbf{Max-Weight Perfect Matching.}  Recently, Chang
et al.~\cite{chang2018dispatch} presented a
$\nicefrac{1}{2}$-competitive algorithm for the \emph{maximum}-weight
perfect matching problem in the i.i.d.\ setting. We show that our
algorithm is versatile, and that a small change to our algorithm gives
us a maximization variant matching this factor of $\nicefrac12$. Our approach differs from that of \cite{chang2018dispatch}, in that we match an arriving request based on the realization of free servers, while they do so based on the ``expected realization''. See Appendix~\ref{sec:max-weight-matching} for details.

\subsection{Our Techniques}\label{sec:technqiues}

Both Theorems~\ref{thm:main} and~\ref{thm:tree} are achieved by the same
algorithm. The first observation guiding this algorithm is that
we may assume that the distribution $\calD$ of request locations is just
the uniform distribution on the server locations. (In
Appendix~\ref{sec:transshipment} we show how this assumption can be
removed with a constant factor loss in the competitiveness.)  Our
algorithm is inspired by the following two complementary
consequences of the uniformity of~$\calD$.
\begin{itemize}
\item Firstly, each of the $n-t+1$ free servers' locations at time $t$
  are equally likely to get a request in the future, and as such they
  should be left unmatched with equal probability. Put otherwise, we
  should match to them with equal probability of $1/(n-t+1)$.  However,
  matching \emph{any} arriving request to any free server with
  probability $1/(n-t+1)$ is easily shown to be a bad choice.
\item So instead, we rely on the second observation: the $t^\textrm{th}$
  request is equally likely to arrive at each of the $n$ server
  locations. This means we can couple the matching of free server
  locations with the location of the next request, to guarantee a
  marginal probability of $1/(n-t+1)$ for each free server to be matched
  at time $t$.
\end{itemize}
Indeed, the constraints that each location is matched at time $t$ with
probability $1/n$ (i.e., if it arrives) and each of the free servers are
matched with marginal probability $1/(n-t+1)$ can be expressed as a
\emph{bipartite flow} instance, which guides the coupling used by the
algorithm. Loosely speaking, our algorithm is fairly intuitive. It finds
a min-cost fractional matching between the current open server locations
and the expected arrivals, and uses that to match new requests. The
challenge is to bound the competitive ratio---in contrast to previously
used approaches (for the maximization version of the problem) it
does not just try to match vertices using a fixed template of choices,
but rather dynamically recomputes a template after each arrival.

A major advantage of this approach is that we understand the distribution of the open servers. 
We maintain the invariant that after $t$ steps, the set of free servers form a
uniform random $(n-t)$-subset of $[n]$---the randomness being over our
choices, and over the randomness of the input. This allows us to relate
the cost of the algorithm in the $t^\textrm{th}$ step to the expected
cost of this optimal flow between the original $n$ points and a
uniformly random subset of $(n-t)$ of these points. The latter expected
cost is just a statistic based on the metric, and does not depend on our
algorithm's past choices. For paths and trees, we bound this quantity
explicitly by considering the variance across edge-cuts in the 
tree---this gives us the proof of Theorem~\ref{thm:tree}.

Since general metrics do not have any usable cut structure, we need a
different idea for Theorem~\ref{thm:main}. We show that tree-embedding
results can be used either explicitly in the algorithm or just implicitly 
in the proof, but both give an $O(\log n)$ loss. To avoid this loss, we 
use a different
balls-and-bins argument to improve our algorithm's competitiveness to
$O((\log \log n))^2)$. 
In particular, we provide better bounds on our algorithm's per-step cost 
in terms of $\E[OPT]$ and the expected load of the $k$ most loaded bins in
a balls 
and bins process, corresponding to the number of requests in the $k$ most frequently-requested servers. 
Specifically, we show that $\E[OPT]$ is bounded in terms of the expected \emph{imbalance} between the number of requests and servers in these top $k$ server locations.
Coupling this latter uniform $k$-tuple with the uniform $k$-tuple of free servers left by our algorithm,
we obtain our improved bounds on the per-step cost of our algorithm in terms of $\E[OPT]$
and these bins' load, from which we obtain our improved $O((\log \log
n)^2)$ competitive ratio.
Interestingly, combining both balls and bins and tree embedding bounds for the per-step cost of step $k$ (appealing to different bounds for different ranges of $k$) gives us a further improvement: we
prove that our algorithm is $O((\log \log \log n)^2)$ competitive.

\subsection{Further Related Work}
\label{sec:further-related-work}

I.i.d.~stochastic arrivals have been studied for various online
problems, e.g., for Steiner tree/forest~\cite{GGLS08}, set
cover~\cite{GGLMSS13}, and k-server~\cite{DEHLS}. Closer to our work, 
stochastic arrivals have been widely studied in the online matching
literature, though so far mostly for maximization variants. Much of this
work was motivated by applications to online advertising, for which the
worst-case optimal $(1-\nicefrac1e)$-competitive ratios
\cite{karp1990optimal,mehta2007adwords,aggarwal2011online} seem
particularly pessimistic, given the financial incentives involved and
time-learned information about the distribution of
requests. Consequently, many stochastic arrival models have been
studied, and shown to admit better than $1-\nicefrac1e$ competitive
guarantees.  The stochastic models studied for online matching and
related problems, in increasing order of attainable competitive ratios,
include random order
(e.g., \cite{goel2008online,karande2011online,mahdian2011online}),
unknown i.i.d.---where the request distribution is unknown---(e.g.,
\cite{devanur2012asymptotically,mirrokni2012simultaneous}), and known
i.i.d.~(e.g.,
\cite{feldman2009online,bahmani2010improved,brubach2016new}).
Additional work has focused on interpolating between adversarial and
stochastic input (e.g.,
\cite{essfandiari2015online,mahdian2007allocating}). See Mehta's survey
\cite{mehta2013online} and recent work
\cite{cohen2018randomized,huang2018match,huang2018online,huang2019tight,gamlath2019online,naor2018near}
for more details. The long line of work on online matching, both under
adversarial and stochastic arrivals, have yielded a slew of algorithmic
design ideas, which unfortunately do not seem to carry over to
minimization problems, nor to perfect matching problems.

As mentioned above, the only prior work for stochastic online matching with
minimization objectives was the random order arrival result of
Raghvendra \cite{raghvendra2016robust}.  We are hopeful that our work
will spur further research in online minimum-cost perfect matching under
stochastic arrivals, and close the gap between our upper bounds and the
(trivial) lower bounds for the problem.
\section{Our Algorithm}
\label{sec:our-algorithm}

In this section we present our main algorithm, together with some of its
basic properties. Throughout the paper we assume that the distribution
over request locations is uniform over the $n$ servers' locations. We
show in \Cref{sec:transshipment} that this assumption is WLOG: it
increases the competitive ratio by at most a constant. 
In particular, we show the following.
\begin{restatable}{lem}{transshipment}\label{transshipment}
	Given an $\alpha$-competitive algorithm $\ALG_{\mathcal{U}}$  for the \emph{uniform} distribution over server locations, $\mathcal{U}$,
	we can construct a $(2\alpha+1)$-competitive algorithm $\ALG_{\mathcal{D}}$ for any 
	distribution $\mathcal{D}$. 
\end{restatable}
Focusing on the uniform distribution over server locations, our algorithm is loosely the following: in each round of the
algorithm, we compute an optimal fractional matching between remaining
free servers and remaining requests (in expectation). Now when a new
request arrives, we just match the newly-arrived request according to
this matching.

\subsection{Notation}\label{sec:notation}

Our analysis will consider $k$-samples from the set $S = [n]$ both with
and without replacement. We will set up the following notation to distinguish them:
\begin{itemize}
\item Let $\calI_k$ be the distribution over $k$-sub-multisets of
  $S=[n]$ obtained by taking $k$ i.i.d.\ samples from the uniform
  distribution over $S$.  (E.g., $\calI_n$ is the request set's distribution.)
\item Let $\calU_k$ be the distribution over $k$-subsets of $S$ obtained
  by picking a uniformly random $k$-subset from $\binom{S}{k}$.
\end{itemize}
In other words, $\calI_k$ is the distribution obtained by picking $k$
elements from $S$ uniformly \emph{with} replacement, whereas $\calU_k$
is \emph{without} replacement.
  
For a sub-(multi)set $T\subseteq S$ of servers, let $M(T)$ denote the
optimal fractional min-cost $b$-matching in the bipartite graph induced
between $T$ and the set of all locations $S$, with overall unit capacity
on either side. That is, the capacity for each node in $T$ is $1/|T|$
and the capacity for each node in $S$ is $1/n$. So, if we denote by
$d_{i,j}$ the distance between locations $i$ and $j$, we let $M(T)$
correspond to the following linear program.
\begin{align*}
	M(T) := \min & \sum_{i\in T,j\in S}  d_{i, j} \cdot x_{i,j} \tag{$M(\cdot)$}\\
	\text{s.t. }  \sum_{j \in S} x_{i,j} & = \tsty \frac{1}{|T|} \qquad \forall i\in T \\
	\sum_{i \in T} x_{i,j} & = \tsty \frac{1}{n}  \qquad\,\,\,\, \forall j \in S \\
	x &\geq 0
\end{align*}
We emphasize that in the above LP,
several servers in $S$ (and likewise in $T$) may happen to be at the same point in the metric
space, and hence there is a separate constraint for each such point
$j$ (and likewise $i$). Slightly abusing notation, we let $M(T)$ denote both the LP and
its optimal value, when there is no scope for confusion.

\subsection{Algorithm Description}
\label{sec:algorithm}

The algorithm works as follows: at each time $k$, if $S_k\subseteq S$ is
the current set of free servers, we compute the fractional assignment
$M(S_k)$, and assign the next request randomly according to it. As
argued above, since each free server location is
equally likely to receive a request later (and therefore it is worth not
matching it), it seems fair to leave each free server unmatched
with equal probability. Put otherwise, it is only fair to match each of
these servers with equal probability. Of course, matching any arriving
request to a free server chosen uniformly at random can be a terrible
strategy. In particular, it is easily shown to be
$\Omega(\sqrt{n})$-competitive for $n$ servers equally partitioned among
a two-point metric. 
Therefore, to obtain good expected matching cost, we should bias
servers' matching probability according to the arrived request, and in particular we should bias it according to $M(S_k)$.
This intuition guides our algorithm \algoname,
and also inspires its name.
\begin{algorithm}[h] 
	\caption{\algoname}
	\label{alg:main-algo}
	\begin{algorithmic}[1]
	    \State $S_n \gets S.$  \Comment{$S_k$ is the set of free servers, with $|S_k| = k$.}
	    \For{time step $k = n, n - 1, \cdots, 1$}
	    \State compute optimal fractional matching $M(S_k)$, denoted by $x^{S_k}$.
	    \Event{arrival of request $r_k=r$}
	    \State randomly choose server $s$ from $S_k$, where $s_i$ is chosen w/prob.\ $p_i = n \cdot x_{s_{i},r}^{S_{k}}$.
	    \State assign $r$ to $s$.
	    \EndEvent
	    \State $S_{k - 1} \gets S_{k} \setminus \{s\}$.
	    \EndFor
	\end{algorithmic}
\end{algorithm}

A crucial property of our algorithm is that the set $S_k$ of free servers at each time $k$
happens to be a uniformly random $k$-subset of $S$. Recall that
\algoname assigns each arriving request according to the assignment
$M(S_k)$. This means that to analyze the algorithm, it suffices to
relate the optimal assignment cost $\OPT$ to the optimal assignment
costs for uniformly random subsets $S_k$, as follows.

\begin{lem}\label{structure}(Structure Lemma)
	For each time $k$, the set $S_k$ is a uniformly-drawn $k$-subset of $S$; i.e., $S_k\sim \calU_k$. Consequently, the algorithm's cost is $$\E[ALG] = \sum_{k = 1}^{n}\E_{S_k \sim \;\calU_k}[M(S_k)].$$
\end{lem}
\begin{proof}
	The proof of the first claim is a simple induction from $n$ down to $1$. The base case of
	$S_n$ is trivial. For any $k$-subset
	$T = \{s_1, \cdots, s_k\} \subseteq S$, 
	\begin{align*}
		\Pr \left[S_k = T  \right] 
		&= \sum_{s \in S\setminus T} \Pr\left[S_{k + 1} = T \cup
		\{s\}  \right] \cdot \Pr\left[ r_{k+1} \text{ assigns to } s \mid S_{k + 1} = T \cup \{s\}\right]\\
		&=(n - k)\cdot \frac{1}{\binom{n}{k + 1}}\cdot \frac{1}{k + 1}
		= \frac{1}{\binom{n}{k}},
	\end{align*}
	where the second equality follows from induction and the fact that
	\begin{gather*}
		\Pr \left[ r_{k+1} \text{ assigned to } s \mid S_{k + 1} =
		T \cup \{s\}\right] = \sum_{r \in S} x^{S_{k+1}}_{s,r} = \frac{1}{k + 1}.
	\end{gather*}
	To compute the algorithm's cost, we consider some set $S_k=T$ of $k$ free servers. 
	Since the request $r_k=r$ is chosen with probability $1/n$, following which we 
	match it to some free server $s\in S_k$ with probability $n\cdot x^{S_k}_{s,r}$, 
	we find that the next edge matched by the algorithm has expected cost
	$$\E[d_{s,r_k} \mid S_k = T] = \sum_r \frac{1}{n}\cdot \sum_{s\in T} n\cdot x^T_{s,r} \cdot d_{s,r} = M(T).$$ 
	Therefore, the expected cost of the algorithm is indeed
	\begin{align*}
		\E[ALG] &= \sum_{k = 1}^{n}\E[d_{s, r_k}]
		= \sum_{k = 1}^{n} \sum_{T\in \binom{S}{k}} \Pr_{S_k \sim \calU_k}[S_k = T]\cdot \E[d_{s, r_k} \mid S_k = T] \\
		&= \sum_{k = 1}^{n} \sum_{T\in \binom{S}{k}} \Pr_{S_k \sim \calU_k}[S_k = T]\cdot M(T) 
		= \sum_{k = 1}^{n}\E_{S_k \sim \calU_k}[M(S_k)]. \qedhere
	\end{align*}
\end{proof}

The structure lemma implies that we may assume from now on that the set
of free servers $S_k$ is drawn from $\calU_k$. In what follows, unless
stated otherwise, we have $S_k\sim \calU_k$. More importantly, 
\Cref{structure} implies that to bound our algorithm's competitive ratio by $\alpha$, it
suffices to show that $\sum_k \E[M(S_k)] \leq \alpha\cdot \E[\OPT]$. This is
exactly the approach we use in the following sections.

\section{Bounds for General Metrics}
\label{sec:general}

In \Cref{sec:trees} we will show that algorithm \algoname is $O(1)$-competitive for line metrics (and 
more generally tree metrics), by relying on variance bounds of the number of matches across tree edges in $OPT$ and $M(S_k)$, our algorithm's guiding LP.
For general metrics, if we first embed the metric in a low-stretch tree metric \cite{fakcharoenphol2004tight}
(blowing up the expected cost of $\E[\OPT]$ by $O(\log n)$) and run algorithm
\algoname on the obtained metric, we immediately obtain an $O(\log
n)$-competitive algorithm.
In fact, explicitly embedding
the input metric in a tree metric is not necessary in order to obtain this result using our algorithm. 
By relying on an \emph{implicit} tree embedding, we obtain the following lemma (mirroring the variance-based bound underlying our result for tree metrics). This lemma's proof is deferred to
\Cref{sec:implicit-tree}.

\begin{restatable}{lem}{TreeEmbed}
	\label{tree-embedding-bound}
	$\E_{S_k\sim \calU_k}[M(S_k)] \leq \frac{O(\log n)}{\sqrt{nk}}\cdot \E[\OPT]$.
\end{restatable}

Summing over all values of $k \in [n]$, we find that \algoname is
$O(\log n)$-competitive on general metrics. While this bound is no
better than that of Raghvendra's $t$-net algorithm for
random order arrival \cite{raghvendra2016robust} (and therefore for i.i.d arrivals), the result
will prove useful in our overall bound for our
algorithm. 
In Sections \ref{sec:poisson} and \ref{sec:loglog}, we use a
different balls-and-bins argument to decrease our bounds on the
algorithm's competitive ratio considerably, to $O((\log \log n))^2)$, by considering the imbalance between number of requests and servers 
in the top $k$ most requested locations. (The former quantity corresponds to the load of the $k$ most loaded bins in a balls and bins process -- motivating our interest in this process.) Finally, in
\Cref{sec:logloglog}, we combine this improved bound with the one from
Lemma~\ref{tree-embedding-bound}, summing different bounds for different ranges of $k$, to prove our main result: an 
$O((\log \log \log n)^2)$ bound for our algorithm's competitive ratio.

\subsection{Balls and Bins: The Poisson Paradigm}\label{sec:poisson}

For our results, we need some technical facts about the classical
balls-and-bins process. 

The following standard lemma from \cite[Theorem 5.10]{mitzenmacher2005probability} allows us to use the
Poisson distribution to approximate monotone functions on the bins.  For
$i \in [n]$, let $X_i^{m}$ be a random variable denoting the number of
balls that fall into the $i^{th}$ bin, when we throw $m$ balls into $n$
bins. Let $Y_i^{m}$ be independent draws from the Poisson distribution
with mean $m/n$.
  
\begin{lem}\label{lem:poisson_approx}
	Let $f(x_1, \cdots, x_n)$ be a
	non-negative function such that $\E[f(X_1^{m}, \cdots, X_{n}^{m})]$ is
	either monotonically increasing or decreasing with $m$,
	then
	$$\E[f(X_1^{m}, \cdots, X_{n}^{m})] \leq 2\cdot \E[f(Y_1^{m}, \cdots, Y_{n}^{m})].$$
\end{lem}

A classic result states that for $m=n$ balls, the maximum bin load is $\Theta(\log n / \log \log n)$~w.h.p. (see e.g., \cite[Lemmas 5.1, 5.12]{mitzenmacher2005probability}). The following lemma is a partial generalization of this result. 
Its proof, which relies on the Poisson approximation of \Cref{lem:poisson_approx}, is deferred to \Cref{sec:app-proofs}.
\begin{restatable}{lem}{topkballsandbins}\label{lem:top-k-balls-and-bins}
	Let $n$ balls be thrown into $n$ bins, each ball thrown independently and
	uniformly at random. Let $L_j$ be the load of the $j^{th}$ heaviest
	bin, and $N_k := \sum_{j \leq k} L_j$ be the number of balls in the $k$
	most loaded bins. There exists a constant $C_0>0$ such that for any
	$k \leq C_0 n$,
	$$\E[N_k] \geq \Omega\left(k\cdot \frac{\log (n/k)}{\log \log (n/k)}\right).$$
\end{restatable}

In the next lemma, whose proof is likewise deferred to \Cref{sec:app-proofs}, we rely on a simple Chernoff bound to give a weaker lower bound for $\E[N_k]$ that holds for
all $k \leq n/2$. 
\begin{restatable}{lem}{topkchernoff} \label{lem:top-k-bins-and-balls1}
  For sufficiently large $n$ and any $k \leq n/2$, we have $\E[N_k] \geq 1.5k$.
\end{restatable}

\subsection{Relating Balls and Bins to Stochastic Metric Matching}
\label{sec:loglog}

We now bound the expected cost incurred by \algoname at time $k$ by
appealing to the above balls-and-bins argument; this will give us our stronger bound of $O((\log \log n)^2)$. Specifically, we will derive
another lower bound for $\E[\OPT]$ in terms of
$\E_{S_k\sim \calU_k}[M(S_k)]$. In our bounds we will partition the
probability space $\calI_{n}$ (corresponding to $n$ i.i.d.~requests)
into disjoint parts, based on $\Tk$, the top $k$ most frequently
requested locations (with ties broken uniformly at random). By
symmetry,
$\Pr[\Tk = T] = 1/ \binom{n}{k}$ for all $T \in \binom{S}{k}$. By coupling $\Tk$ with $\calU_k$,
we will lower-bound $\E[OPT]$ by $\E_{S_k\sim \calU_k}[M(S_k)]$ times $\E[N_k]-k$, the expected imbalance between number of requests and servers in $\Tk$. Here $\E[N_k]$ is the expected occupancy of the $k$ most
loaded bins in the balls and bins process discussed in
\Cref{sec:poisson}.

To relate $\E[OPT \mid \Tk = S_k]$ to $M(S_k)$, 
we will bound both these quantities by the 
cost of a min-cost perfect $b$-matching between $S_k$ and $S\setminus S_k$; i.e., each vertex $v$ has some (possibly fractional) demand $b_v$ which is the extent to which it must be matched. 
To this end, we need the following simple lemma, which asserts that for any min-cost metric $b$-matching instance, there exists an optimal solution which matches co-located servers and requests maximally.
We defer the lemma's proof, which follows from a local change argument and triangle inequality, to \Cref{sec:app-proofs}. 
\begin{restatable}{lem}{MatchToSelf}
	\label{lem:match-to-self}
	Let $\calI$ be a fractional min-cost bipartite metric $b$-matching instance, 
	with demand $\ell_i$ and $r_i$ for the servers and requests at location $i$.
	Then, there exists an optimal solution $x$ for $\calI$ with $x_{ii} = \min\{\ell_i,r_i\}$ for every point $i$ in the metric.
\end{restatable}

We are now ready to prove our main technical lemma, lower-bounding $\E[OPT \mid \Tk = S_k]$ in terms of $M(S_k)$ and the imbalance between number of requests of the $k$ most requested locations, $N_k$, and the number of servers in those locations.
\begin{lem}\label{balls-and-bins-bound-coupling}
	For all $k < n$ and $S_k\in \binom{S}{k}$, we have $\E[OPT \mid \Tk = S_k] \geq (\E[N_k] - k)\cdot M(S_k).$
\end{lem}
\begin{proof}
Applying \Cref{lem:match-to-self} to $M(S_k)$, we find that the optimal value of $M(S_k)$ is equal to that of a min-cost bipartite perfect $b$-matching instance with left vertices associated with $S_k$, each with demand $\frac{1}{k}-\frac{1}{n}$, and right vertices associated with $S\setminus S_k$, each with demand $\frac{1}{n}$.

We now turn to the meat of the proof -- lower bounding $\E[OPT \mid \Tk = S_k]$. In particular, we will lower bound $\E[OPT \mid \Tk = S_k]$ by a 
min-cost bipartite perfect $b$-matching instance with left and right vertices as 
above (i.e., $S_k$ and $S\setminus S_k$, respectively), but with uniform demands on both sides of at least $(\E[N_k]-k)/k$ and $(\E[N_k]-k)/(n-k)$, respectively. That is, the 
biregular min-cost bipartite $b$-matching whose cost $C$ we showed lower bounds $M(S_k)$, but 
scaled by an $f\geq \frac{(\E[N_k]-k)}{k\cdot (1/k-1/n)}$ factor.	Before proving this lower bound on $\E[OPT \mid \Tk = S_k]$, we note that it implies our desired bound, as
\begin{equation*}
	\E[\OPT \mid \Tk = S_k] \geq \frac{(\E[N_k] - k)}{k\cdot (1/k - 1/n)}\cdot C > (\E[N_k] - k) \cdot C = (\E[N_k] - k) \cdot M(S_k).
\end{equation*}
It remains to lower bound $\E[OPT \mid \Tk = S_k]$ in terms of such a biregular $b$-matching instance.

For the remainder of this proof, for notational simplicity we denote by  $\Omega$ the probability space induced by conditioning on the event $\Tk = S_k$. 
To lower bound $\E_\Omega[OPT]$, we will provide a fractional perfect matching $\vec{x}$ of the 
expected instance (in $\Omega$), and show that 
$\E_\Omega[OPT] \geq \sum_{ij} d_{ij}\cdot x_{ij}$, while $\sum_{j\in S\setminus S_k} x_{ij} \geq (\E[N_k]-k)/k$ 
for all $i\in S_k$ and $\sum_{i\in S} x_{ij} \geq (\E[N_k]-k)/(n-k)$ for all $j\in S\setminus S_k$. 
Consequently, focusing on edges $(i,j)\in S_k\times (S\setminus S_k)$, we find that the min-cost biregular bipartite perfect $b$-matching above lower bounds 
$\sum_{i\in S_k,j\in S\setminus S_k} d_{ij}\cdot x_{ij}\leq \sum_{ij} d_{ij}\cdot x_{ij} \leq \E_\Omega[OPT]$. 
We now turn to producing an $\vec{x}$ satisfying our desired properties. 

For any two locations $i,j\in S$, we let $(i,j)\in OPT$ indicate that a request in location $i$ 
is served by the server in location $j$.
Let $p_{ij} := \Pr_{\Omega}[(i,j)\in OPT]$. We will show how small modifications to $\vec{p}$ will yield a fractional perfect matching $\vec{x}$ as discussed in the previous paragraph. Let $Y_i$ be the number of requests at server $i$.
By \Cref{lem:match-to-self}, we know that $(i,i)\in OPT \iff Y_i\geq 1$. So, $p_{ii} = \Pr_{\Omega}[Y_i\geq 1]$. 
Consequently, if we let $\Delta_{in}(j) := \sum_{j'\in S\setminus\{j\}} p_{j'j}$ and $\Delta_{out}(j) := \sum_{j'\in S\setminus\{j\}} p_{jj'}$, 
we have by \Cref{lem:match-to-self} that $\Delta_{in}(j) = \Pr[Y_i \geq 1]$ and 
$\Delta_{out}(i) = \E[(Y_i - 1)^+]$ for all $i\in S$. (As usual, $x^+=\max\{x,0\}$.) 
Consequently, $\Delta_{in}(j) = \Delta_{in}(j')$  and $\Delta_{out}(j) = \Delta_{out}(j')$ for 
all $j,j'\in S\setminus S_k$, as $[Y_j \mid \Omega]$ and $[Y_j' \mid \Omega]$ are 
identically distributed. Moreover, as 
$\sum_{j\in S\setminus S_k} \left(\Delta_{in}(j) - \Delta_{out}(j)\right) = N_k - k \geq 0$, we find that $\Delta_{in}(j) - \Delta_{out}(j)\geq 0$ for all $j\in S\setminus S_k$.
Now, suppose $Y_i\geq 1$ for all $i\in S_k$ (conditioning on the complementary event is similar), we have by \Cref{lem:match-to-self} that $p_{ji}=0$ for all $i\in S_k$ and $j\in S\setminus\{i\}$. Moreover, by symmetry we have $\Delta_{out}(i) = (\E[N_k] - k)/k$ for all $k$ locations $i\in S_k$. We now show how to obtain from $\vec{p}$ a fractional matching $\vec{x}$ between $S_k$ and $S\setminus S_k$ of no greater cost than $\vec{p}$, such that $p_{jj'}=0$ for all $j\neq j'\in S\setminus S_k$ and such that the values $\Delta_{in}(j) - \Delta_{out}(j)$ are unchanged for all $j\in S$. Consequently, all (simple) edges in the support of $\vec{x}$ go between $S_k$ and $S\setminus S_k$, and $\Delta_{out}(i)= (\E[N_k] - k)/k$ for all $i\in S_k$ and $\Delta_{in}(j) = (\E[N_k] - k)/(n-k)$ for all $j\in S\setminus S_k$, yielding our desired lower bound on $\E_\Omega[OPT]$ in terms of a biregular bipartite $b$-matching instance.

We start by setting $\vec{x}\leftarrow \vec{p}$. 
While there exists a pair $j\neq j'\in S\setminus S_k$ with $x_{j'j}>0$, we pick such a pair. As $\Delta_{in}(j) - \Delta_{out}(j) \geq 0$, there must also be some flow coming into $j$. We follow a sequence of edges $j_1 \leftarrow j_2 \leftarrow j_3 \leftarrow \dots$ with each $j_r\in S\setminus S_k$ and with $x_{j_rj_{r-1}} > 0$ until we either repeat some $j_r \in S\setminus$ or reach some $j_r$ with $x_{ij_r} 0$ for some $i\in S$. (Note that one such case must happen, as $\Delta_{in}(j) - \Delta_{out}(j) \geq 0$ for all $j\in S\setminus S_k$.) If we repeat a vertex, $j_r$, we only consider the sequence of nodes given by the obtained cycle, $j_1 \leftarrow j_2 \leftarrow j_3 \dots \leftarrow j_r = j_1$.
Let $\epsilon = \min_{r} x_{j_rj_{r-1}}$ be the smallest $x_{jj'}$ in our trail. If we repeated a vertex, we found a cycle, and we decrease $x_{jj'}$ by $\epsilon$ for all consecutive $j,j'$ in the cycle. If we found some $i\in S$ and $x_{ij_r} > 0$, we decrease all $x_{jj'}$ values along the path (including $x_{ij_r}$) by $\epsilon$ and increase $x_{ij_1}$ by $\epsilon$. In both cases, we only decrease the cost of $\vec{x}$ (either trivially, or by triangle inequality) and we do not change $\Delta_{in}(j) - \Delta_{out}(j)$ for any $j\in S$, while decreasing $\sum_{j\neq j'\in S\setminus S_k} x_{jj'}$. As the initial $x$-values are all rational, repeating the above terminates, with the above sum equal to zero, which implies a biregular fractional solution $\vec{x}$ as required. The lemma follows.
\end{proof}
 	
Coupling the distribution of $\Tk$ and the set of $k$ free servers, we obtain the following.
\begin{lem}\label{balls-and-bins-bound-raw}
   	$\E_{S_k\sim \calU_k}[M(S_k)] \leq \E[\OPT] / (\E[N_k] - k)$. 
\end{lem}
\begin{proof}
 	Taking expectations over $S_k\sim \calU_k$, we obtain our claimed bound.
	\begin{align*}
	\E_{S_k\sim \calU_k} [M(S_k)] &= \sum_{S_k\in \binom{S}{k}}\frac{1}{\binom{n}{k}}\cdot M(S_k) & \textrm{defn.\ of } \calU_k \\
	&\leq \sum_{S_k\in \binom{S}{k}} \frac{1}{\binom{n}{k}}
    \frac{1}{(\E[N_k] - k)}\cdot \E[\OPT \mid \Tk = S_k] & \textrm{\Cref{balls-and-bins-bound-coupling}} \\
	&= \frac{1}{(\E[N_k] - k)}  \cdot \E[\OPT]. &
	\Pr[\Tk = S_k] = \frac{1}{\binom{n}{k}}.& \qedhere
	\end{align*}
\end{proof}

Plugging in the lower bounds of Lemmas \ref{lem:top-k-balls-and-bins} and \ref{lem:top-k-bins-and-balls1} for the top $k$ most loaded bins' loads, $\E[N_k]$, 
we obtain the following bounds on \algoname's per-step cost in terms of $\E[OPT]$.
\begin{lem}\label{balls-and-bins-bound}
   	For $C_0$ a constant as in \Cref{lem:top-k-balls-and-bins},  there exists a constant $C$ such that
   	\begin{equation*}
   	\E_{S_k\sim \calU_k}[M(S_k)] \leq \begin{cases}
   	C\cdot \frac{\log\log(n/k)}{k\log(n/k)}\cdot \E[\OPT] & \text{ if } k < C_0 n \\
   	\frac{2}{k}\cdot\E[\OPT] & \text{ if } C_0 n \leq k \leq n/2.
   	\end{cases} 
   	\end{equation*}
\end{lem}

The following lemma allows us to leverage \Cref{balls-and-bins-bound}, as it allows us to focus on $\E_{S_k\sim \calU_k}[M(S_k)]$ for $k\leq n/2$. 
Its proof relies on our characterization of $M(S_k)$ in terms of a balanced $b$-matching instance between $S_k$ and $S\setminus S_k$ as in 
the proof of \Cref{balls-and-bins-bound-coupling}, which implies that $M(S_k) \leq M(S_{n-k})$ for all $k\leq n/2$. Its proof is deferred to \Cref{sec:app-proofs}.
\begin{restatable}{lem}{sumhalf}\label{lem:sum-only-half}
	$\sum_{k=1}^n \E_{S_k\sim \calU_k}[M(S_k)] \leq 2\cdot \sum_{k=1}^{n/2} \E_{S_k\sim \calU_k}[M(S_k)]$.
\end{restatable} 

Using our upper bound on $\E_{S_k\sim \calU_k}[M(S_k)]$ of \Cref{balls-and-bins-bound} and summing the two ranges of $k\leq n/2$ in \Cref{lem:sum-only-half} we find that \algoname is $O((\log \log n)^2)$ competitive. We do not elaborate on this here, as we obtain an even better bound in the following section.

\subsection{Our Main Result}
\label{sec:logloglog}

We are now ready to prove our main result, by combining our per-step cost bounds given by our balls and bins argument (\Cref{balls-and-bins-bound}) and  our implicit tree embedding argument (\Cref{tree-embedding-bound}).
\begin{thm}\label{iid-general-metrics}
  Algorithm \algoname is $O((\log \log \log n)^2)$-competitive for the
  online bipartite metric matching problem under \IID arrivals on
  general metrics.
\end{thm}
\begin{proof}
	By the structure lemma (\Cref{structure}) and \Cref{lem:sum-only-half}, we have that
\begin{equation}\label{half-range}
	\E[ALG] = \sum_{k=1}^n \E_{S_k\sim \calU_k}[M(S_k)] \leq 2\cdot \sum_{k=1}^{n/2} \E_{S_k\sim \calU_k}[M(S_k)].
\end{equation}
We use the three bounds from \Cref{tree-embedding-bound} and
\Cref{balls-and-bins-bound} for different ranges of $k$ to
bound the above sum. Specifically, by relying on \Cref{tree-embedding-bound} for $k\leq n/\log^2n$, we have that
\begin{align*}
	\sum_{k=1}^{n/\log^2n} \E_{S_k\sim \calU_k}[M(S_k)] & \leq \sum_{k=1}^{n/\log^2n} \frac{O(\log n)}{\sqrt{nk}}\cdot \E[OPT] \\
	& \leq O\left(\sqrt{\frac{n}{\log^2n}} \cdot \frac{\log n \cdot \E[OPT]}{\sqrt n}\right) = O(1)\cdot \E[OPT].
\end{align*}
Next, by the first bound of \Cref{balls-and-bins-bound} applied to $k\in [n/\log^2n, C_0n]$, we have that
\begin{align*}
	\sum_{k=n/\log^2n}^{C_0 n} \E_{S_k\sim \calU_k}[M(S_k)] & \leq \sum_{k=n/\log^2n}^{C_0 n} \frac{O(\log \log (n/k))}{k\cdot \log (n/k)}\cdot \E[OPT] \\
	& \leq O\left(-(\log \log (n/k))^2 \Big|_{n/\log^2n}^{C_0 n}\right)\cdot \E[OPT]\\
	& = O((\log \log \log n)^2)\cdot \E[OPT].
\end{align*}
Finally, by the second bound of \Cref{balls-and-bins-bound} applied to $k\geq C_0n$, we have that
\begin{align*}
	\sum_{k=C_0 n}^{n/2} \E_{S_k\sim \calU_k}[M(S_k)] & \leq \sum_{C_0 n}^{n/2} \frac{2}{k}\cdot \E[OPT] \leq O\left(\log\left(\frac{n/2}{C_0 n}\right)\right)\cdot \E[OPT] = O(1)\cdot \E[OPT].
\end{align*}	 

Combining all three bounds with \Cref{half-range}, the theorem follows.
\end{proof}
\section{A Simple $O(1)$ Bound for Tree Metrics}
\label{sec:trees}

In this section we show the power of the structure lemma, by 
analyzing \algoname on tree metrics.
Recall that a \emph{tree metric} is defined by shortest-path distances
in a tree $T = (V,E)$, with edge lengths $d_e$.  By adding zero-length
edges, we may assume that the tree has $n$ leaves, and that servers are
on the leaves of the tree. For any edge $e$ in the tree, deleting this
edge creates two components $T_1(e)$ and $T_2(e)$; denote by $T_1(e)$ the component with fewer servers/leaves. Let $n_e$
denote the number of leaves on this smaller side, 
$T_1(e)$. Hence $n_e \leq n/2$ for all edges $e$.

We now lower bound $\E[\OPT]$, by considering the mean average deviation
of the number of requests which arrive in $T_1(e)$ for each edge $e$.
\begin{lem}\label{lem:variance-lower-bound-tree}
	The expected optimal matching cost in a tree metric on $n\geq 2$ vertices is at least
	$\E[\OPT] \geq \frac{1}{2} \cdot \sum_{e \in T}d_e \cdot \sqrt{n_e}$.
\end{lem}
\begin{proof}
Let $X_e$ denote the number of requests that arrive in
the component with fewer leaves, $T_1(e)$. Every matching will 
match at least $|X_e - n_e| = |X_e - \E[X_e]|$ requests across the edge $e$ (with the equality due to the uniform IID arrivals). Summing
over all edges and taking expectations, we find that
\begin{equation}\label{mad-bound}
\E[\OPT] \geq \sum_{e}d_e \cdot \E\big[|X_e - n_e|\big] =
\sum_{e}d_e\cdot \E\big[|X_e - \E
[X_e]|\big].
\end{equation}
It remains to lower bound $\E[|X_e-\E[X_e]|]$, the mean average deviation of $X_e$.
Observe that $X_e \sim \text{Bin}(n, n_e/n)$, with $n_e\in [1,n-1]$. The
following probabilistic bound appears in 
\cite[Theorem~1]{berend2013sharp}:
\begin{cla}
\label{clm:mad}
Let $Y \sim \text{Bin}(n,p)$, with $n\geq 2$ and $p\in [1/n,1- 1/n]$. Then, we have both
\begin{align*}
\E|Y - \E Y| & \geq \STD(Y)/\sqrt{2},\label{mad-small}
\end{align*}
\end{cla}

(Note that convexity implies that $\E|Y-\E Y| \leq \STD(Y)$ holds for all
distributions, so this is a partial converse.)
Applying \Cref{clm:mad}
to our case, where $p=n_e/n\in [1/n,1- 1/n]$, 
\begin{align*}
	\E[|X_e - \E X_e|] & \geq \STD(X_e)/\sqrt{2} = \sqrt{n_e(1-n_e/n)/2} \geq \sqrt{n_e/4}, 
\end{align*}
where the second inequality 
follows from $n_e\leq n/2$. 
Combined with \eqref{mad-bound}, the lemma follows.
\end{proof}

To upper bound $\E[M(S_k)]$, we again consider the mean average
deviation of the number of requests in $T_1(e)$, but this time when drawing
$k$ \emph{i.i.d.}\ samples. First, we need to bound the cost of $M(S_k)$ for
a set $S_k$ resulting from $k$ draws  \emph{without replacement}
by the cost for a multiset obtained by taking $k$
i.i.d.\ draws \emph{with replacement}.
\begin{lem}(Replacement Lemma)\label{repetition-MSk}
	For all $S$ and $k\in [|S|]$, we have $$\E_{S_k \sim \calU_k}[M(S_k)] \leq \E_{S_k \sim \calI_k}[M(S_k)].$$
\end{lem}

We defer the proof of this lemma to \Cref{sec:replacement}, where we prove a
more general statement regarding stochastic convex optimization with constraints and coefficients determined
by elements of a set chosen uniformly with and without
replacement.
 Armed with this lemma, it suffices to bound
$\E_{S_k\sim \calI_k}[M(S_k)]$ from above, which we do in the following.

\begin{lem}\label{lem:upper-bound-tree}
    $\E_{S_k \sim \calI_k} [M(S_k)] \leq  \sum_{e \in T} d_e \cdot \sqrt{n_e/(k n)}.$
\end{lem}

\begin{proof}
	Fix some edge $e$ and let $T_1(e)$ be its smaller subtree, containing $n_e\leq n/2$ leaves. 
    Let $X_e \sim \mathrm{Bin}(k,n_e/n)$ be the random variable denoting the number of servers in
    $T_1(e)$ chosen in $k$ i.i.d samples from $S$. 
    For any given realization of $S_k$ (and therefore of $X_e$) the fractional
    solution to $M(S_k)$ utilizes edges between the different subtrees of $e$
    by exactly
    $|X_e/k - n_e/n|$.  Since this is a tree metric, we have
    \[M(S_k) = \sum_{e\in T} d_e\cdot \left|\frac{X_e}{k} - \frac{n_e}{n}\right| = \sum_{e\in T}d_e\cdot \frac{1}{k}\cdot \left|X_e - \frac{k}{n}\cdot n_e \right| = \sum_{e\in T}d_e\cdot \frac{1}{k}\cdot |X_e - \E[X_e]|. \]
    Taking expectations over $S_k$, and using the fact that the mean
    average deviation is always upper bounded by the standard
    deviation (by Jensen's inequality), we find that indeed
	\begin{align*}
	\E_{S_k\sim \calI_k}[M(S_k)]
	= & \sum_{e\in T}d_e\cdot \frac{1}{k}\cdot \E[|X_e - \E[X_e]|]
	\leq  \sum_{e\in T} d_e\cdot \frac{1}{k} \cdot \STD(X_e) \\
	= & \sum_{e\in T} d_e\cdot \frac{1}{k} \cdot \sqrt{k\cdot \frac{n_e}{n}\left(1-\frac{n_e}{n}\right)} 
	\leq \sum_{e\in T}d_e \cdot \sqrt{\frac{n_e}{k\cdot n}}.\qedhere
	\end{align*}
\end{proof}

Combining the replacement lemma (\Cref{repetition-MSk}) with
Lemmas~\ref{lem:upper-bound-tree} and
\ref{lem:variance-lower-bound-tree}, we obtain the following upper bound
on $\E_{S_k\sim \calU_k}[M(S_k)]$ in terms of $\E[OPT]$.

\begin{lem}\label{lem:MSk-OPT-bound}
	$\E_{S_k\sim \calU_k}[M(S_k)]\leq 2\cdot \frac{\E[OPT]}{\sqrt{nk}}.$
\end{lem}

We can now prove our simple result for tree metrics.

\begin{thm}(Tree Bound)\label{main-tree}
	Algorithm \algoname is $4$-competitive on tree metrics with $n\geq 2$ nodes, if the requests are drawn from the uniform distribution.
\end{thm}
\begin{proof}
	We have by the structural lemma (\Cref{structure}) and \Cref{lem:MSk-OPT-bound} that 
	\begin{align*}
	\E[\ALG] & = \sum_{k=1}^n \E[M(S_k)] 
	\leq \sum_{k=1}^n 2 \cdot \frac{\E[OPT]}{\sqrt{nk}} \\
	& \leq 2 \cdot \frac{\E[OPT]}{\sqrt{n}}\cdot \left(1+\int_{x=1}^n \frac{1}{\sqrt x} dx\right) \leq 4\cdot \E[OPT].\qedhere
	\end{align*}
\end{proof}

The above bound holds for all $n\geq 2$ (for $n=1$ any algorithm is trivially $1$ competitive). For $n$ large, however, our proof yields an improved 
asymptotic bound of $\sqrt{2}\cdot e + o(1) \approx (3.845+o(1))$, by relying on the asymptotic counterpart of \Cref{clm:mad} in
\cite[Corollary~2]{berend2013sharp}, $\E|Y - \E Y| \geq \STD(Y)/(e/2+o(1))$.
Combining \Cref{main-tree} with our transshipment argument (\Cref{transshipment}), we 
obtain a $9$-competitive algorithm under any i.i.d.~distribution on tree metrics on $n\geq 2$ nodes, and even better than $9$-competitive
algorithms for large enough $n$.
\section{Open Questions}

In this work, we presented algorithm \algoname and proved that it is
$O((\log \log \log n)^2)$-competitive for general metrics, and
$9$-competitive for tree metrics. Perhaps the first question is whether
our algorithm (or indeed any algorithm) is $O(1)$ competitive for (known
or unknown) i.i.d arrivals for general metrics. Indeed, we do not know of any instances where Algorithm \algoname's 
performance is worse than $O(1)$ competitive. However, it is not clear how to extend our proofs
to establish an $O(1)$ competitive ratio. 

Another question is the relationship between the known and unknown
i.i.d.\ models and the random order model. The optimal competitive
ratios for the various arrival models for online problems can be sorted
as follows (see e.g. \cite[Theorem 2.1]{mehta2013online})
\begin{align*}
C.R.(Adversarial) & \geq C.R.(Random\ Order) \geq C.R.(Unknown\ IID) \geq C.R.(Known\ IID).
\end{align*}
For the online metric matching problem the best bounds known for the
above are, respectively, $\Theta(n)$
\cite{khuller1994line,kalyanasundaram1993online},
$\Theta(\log n), O(\log n)$ (both \cite{raghvendra2016robust}), and
$O((\log \log \log n)^2)$ (this work).  Given the lower bound of
\cite{raghvendra2016robust}, our work implies that one or both of the
inequalities in
$C.R.(Random\ Order) \geq C.R.(Unknown\ IID) \geq C.R.(Known\ IID)$ is
strict (and asymptotically so). It would be interesting to see which of
these inequalities is strict, by either presenting a
$o(\log n)$-competitive algorithm for unknown i.i.d or a
$\omega((\log \log \log n)^2)$ lower bound for this model.  For the line
metric, given the lower bound of \cite{fuchs2005online}, our work
implies that one of the three inequalities above must be
strict. Understanding the exact relationships between these arrival
models for this simple metric may prove useful in understanding the
relationships between the different stochastic arrival models more
broadly. Moreover, it would be interesting to study these questions for
other combinatorial optimization problems with online stochastic
arrivals. 

\appendix
\section*{Appendix}
\section{Distribution over Server Locations (Transshipment Argument)}\label{sec:transshipment}

In this section, we show that the assumption that the requests are drawn from
$\mathcal{U}$, the uniform distribution over server locations, is without loss
of generality. 
\transshipment*

\begin{proof}
As before, we identify the set of servers $S$ with the $n$ points on the metric and let 
$r_1,\dots,r_n$ be the requests that arrive according to the 
distribution $\mathcal{D}$. Define $p_i := \Pr_{r \sim \mathcal{D}} [r=i]$.  

Consider the linear program defined by the transshipment problem between
the distribution $\mathcal{D}$ to the uniform distribution on the servers $S$.
\begin{align*}
    LP:= \min & \sum_{i,j}  d_{i,j} \cdot x_{i,j} \\
    \text{s.t. }  \sum_{j} x_{i,j} &= p_i \qquad \forall i\in \text{metric} \\
    \sum_{i} x_{i,j} &= \frac{1}{n}  \qquad \forall j\in S \\
    x &\geq 0
\end{align*}
Let $M = n \cdot LP$.  Given a request sequence $\{r_1,\dots,r_n\}$ drawn from
$\mathcal{D}$, we create a coupled sequence $\{\tilde{r}_1,\dots,\tilde{r}_n\}$
by moving an arrived request $r_k$ at server location $j$ to location $i$ in
the metric with probability $x_{i,j}/p_i$ Each server location $j \in S$
appears with probability $\sum_{i} x_{i,j} = \frac{1}{n}$ and hence the
sequence $\{ \tilde{r}_1,\dots,\tilde{r}_n\}$ is distributed according to the
uniform distribution $\mathcal{U}$.  After this move, it matches the request
according to $\ALG_{\mathcal{_{U}}}$. 

We bound this algorithm's cost as follows.  First, the probability of a given request being moved from some location $i$ to
$j$ is precisely $p_i \cdot x_{i,j}/p_i = x_{i,j}$. Summing up over all $i,j$,
the expected movement cost for all $n$ time steps is precisely $M=n\cdot LP$.
Secondly, the expected cost of matching from $\tilde{r}_i$ is precisely
$\E[\ALG_{\mathcal{U}}]$. By the triangle inequality, we can bound the total
cost by the sum of the initial costs and the matching costs according to
$\ALG_{\mathcal{U}}$, yielding the relation
\begin{equation}\label{algd-bound}
    \E[\ALG_{\mathcal{D}}] \leq \E[\ALG_{\mathcal{U}}] + M.
\end{equation}

We use the same coupling as above, but in the other direction to relate
$\OPT_{\mathcal{U}}$ to $M$.  In particular, given a request sequence
$\{r_1,\dots,r_n\}$ drawn from $\mathcal{U}$, we create a coupled sequence
$\{\tilde{r}_1,\dots,\tilde{r}_n\}$ by moving an arrived request $r_k$ at
server location $j$ to location $i$ in the metric with probability $n\cdot
x_{i,j}$.  Now $\Pr[\tilde{r}_k = i] = \frac{1}{n}\cdot \sum_{j} n\cdot x_{i,j}
= \sum_{j} x_{i,j} = p_i$. That is, the resulting distribution is
$\mathcal{D}$.  One way to bound the optimal solution for distribution
$\mathcal{U}$ is to match request $r_k$ to the match of $\tilde{r}_k$. As
before, the expected movement cost to locations
$\{\tilde{r}_1,\dots,\tilde{r}_n\}$ is $M$, and by triangle inequality, we find
that
\begin{equation}\label{optu-bound}
    \E[\OPT_{\mathcal{U}}] \leq \E[\OPT_{\mathcal{D}}] + M.
\end{equation}

We now bound $\E[\OPT_{\mathcal{D}}]$ in terms of $M$. Each location $i$ in the
metric has an expected $np_i$ appearances, who must therefore be matched an
expected $np_i$ many times. Each server, on the other hand, is matched
precisely once in expectation. Therefore, the probabilities $p_{i,j}$ of an
arrival at location $i$ being matched to a server at location $j$ constitute a
feasible solution to $n\cdot LP$, and so must have 
$\sum_{i,j} d_{i,j} \cdot p_{i,j} \geq n\cdot LP = M$. 
Therefore, $\E[\OPT_{\mathcal{D}}]$ satisfies 
\begin{equation}\label{optd-bound}
	\E[\OPT_{\mathcal{D}}] \geq M.
\end{equation}

Combining equations \eqref{algd-bound}, \eqref{optu-bound} and \eqref{optd-bound} with $\ALG_{\mathcal{U}}$'s $\alpha$-competitiveness, we obtain our desired result.
\begin{align*}
    \E[\ALG_{\mathcal{D}}] & \leq \E[\ALG_{\mathcal{U}}] + M  & \text{\Cref{algd-bound}} & \\
     &\leq \alpha \cdot \E[\OPT_{\mathcal{U}}] + M & \ALG_{\mathcal{U}} \text{ is $\alpha$-comp.} & \\
     &\leq \alpha \cdot (\E[\OPT_{\mathcal{D}}] + M) + M & \text{\Cref{optu-bound}} & \\
     &\leq (2\alpha+1) \cdot \E[\OPT_{\mathcal{D}}]. & \text{\Cref{optd-bound}} & \qedhere
\end{align*}
\end{proof}
\section{Stochastic Convex Optimization,\\ with and without Replacement}\label{sec:replacement}

In \Cref{repetition-MSk} we claimed that the expected cost of the linear
program $M(S_k)$ for $S_k$ chosen at random from the $k$-subsets of $S$ is
lower than its counterpart when $S_k$ is obtained from $k$ i.i.d draws from
$S$. More succinctly, we claimed that $\E_{S_k\sim \calU_k}[M(S_k)]\leq
\E_{S_k\sim \calI_k}[M(S_k)]$. In this section we prove a more general claim for any 
linear program (and more generally, any convex program), implying the above. 
Let $S$ be some $n$-element set, and for any multiset $T$ with all its elements taken from $S$, 
let $P(T)$ be the following convex program.
\begin{align*}
	P(T) := \min & f(x,\chi_T) \tag{$P(\cdot)$} \\
	\text{s.t. } & g_i(x,\chi_T) \leq 0 \qquad \forall i\in [m] \\
				 & h_j(x,\chi_T) = 0 \qquad \forall j\in [\ell]
\end{align*}
Here $f(x,\chi_T)$ and all $g_i(x,\chi_T)$ are convex functions and $h_j(x,\chi_T)$ are affine in their arguments $x$ and $\chi_T$, and $\chi_T$ is the incidence vector of the multiset $T$. (That is, for any $s\in S$, we let $\chi_T(s)$ denote the number of appearances of $s$ in $T$.) Note that $M(T)$ defined in \Cref{sec:notation} is a linear program of the above form. As such, the following lemma generalizes -- and implies -- \Cref{repetition-MSk}.

\begin{lem}[Replacement Lemma]\label{replacement-general} For any convex program $P$ as above,
we have 
	$$\E_{S_k \sim \calU_k}[P(S_k)] \leq \E_{S_k \sim \calI_k}[P(S_k)].$$
\end{lem}

\begin{proof}
	Our proof relies on a coupling argument, starting with a refined partition of the probability space of $S_k\sim \calI_k$. 
	This space is partitioned into equiprobable events $A_M$ for each ordered multiset $M$ of size $k$ supported in $S$, corresponding to $M$ being sampled. 
	For each ordered multiset $M$, we denote by $\support(M) := \{s\in S \mid s\in M\}$ the set of elements in $M$. 
	Next, we denote by $\SUP(M):=\{T\in \binom{S}{k} \mid T\supseteq \support(M) \}$ the family
	of $k$-sets which contain $M$'s elements (i.e., supersets of $M$'s support). 
	We will wish to ``equally partition'' the event $A_M$ among the $k$-tuples in $\SUP(M)$. To this end, when $M$ 
	is sampled from $\calI_k$, we roll a $|\SUP(M)|$-sided die labeled by the members of $\SUP(M)$. 
	For any $k$-set $T\in \SUP(M)$, we denote by $A_{M,T}$ the event that $M$ was sampled from $\calI_k$ 
	and the die-roll came out $T$, and 
	for any $k$-tuple $T\in \binom{S}{k}$, we let $A_T := \bigcup_{M}{A_{M,T}}$.
	It is easy to verify that by symmetry we have $Pr[A_T] = 1/\binom{|S|}{k}$ for every $T\in \binom{S}{k}$.
	
	We now wish to couple the above refinement of the probability space of $\calI_k$ and 
	the optimal solution to $P(S_k)$ with their counterpart under $\calU_k$. We will need the following claim.
	
	\begin{cla}\label{exp-apps}
		For all $k$-set $T\in \binom{S}{k}$ and element $s\in T$, we have $\E_{S_k \sim \calI_k}[\chi_{S_k}(s) \mid A_T] = 1.$
	\end{cla}
	\begin{proof}
        By definition, each non-empty $A_{M,T}\subseteq A_T$ satisfies
        $\E_{S_k \sim \calI_k}[\sum_{s\in T} \chi_{S_k}(s) \mid A_{M,T}] = k$, since any ordered
        multiset $M$ of size $k$ with $\SUP(M)\ni T$ has all its elements in
        $T$.  Therefore, taking total expectation over $M$ with $\SUP(M)\ni T$,
        we get $\E_{S_k \sim \calI_k}[\sum_{s\in T} \chi_{S_k}(s) \mid A_{T}] = k$. 
        Therefore, by symmetry, we find that indeed each of the $k$ elements $s\in T$ 
        has $\E_{S_k \sim \calI_k}[\chi_{S_k}(s) \mid A_{T}] = 1$.
	\end{proof}
	
	Now, consider some $k$-set $T\in \binom{S}{k}$. For any ordered multiset of 
	$k$ elements $M$ such that $SUP(M)\ni T$, denote by $x^M \in \arg\min P(M)$ a solution 
	of $P(M)$ of minimum cost. By definition, for each $i\in [m]$ we have that
	$g_i(x^M,\chi_M)\leq 0$ and for each $j\in [\ell]$ we have that $h_j(x^M,\chi_M) = 0$. Therefore, if we let $y^T := \E_{M\sim \calI_k}[x^M \mid A_T]$ be the ``average'' optimal solution for $P(M)$ over all $M$ with $SUP(M)\ni T$, 
	then by Jensen's inequality and convexity of $g_i$, we have that 
	\begin{align*}
		0 & \geq \E_{M\sim \calI_k}[g_i(x^M,\chi_M) \mid A_T] & \textrm{linearity}\\
		& \geq g_i(\E_{M\sim \calI_k}[x^M \mid A_T], \E_{M\sim \calI_k}[\chi_M \mid A_T]) & \textrm{Jensen's Ineq.}\\
		& = g_i(y^T, \chi_T). & \textrm{\Cref{exp-apps}}
	\end{align*}
	Similarly, we have that $h_j(y^T,\chi_T) = \E_{M\sim \calI_k}[h_j(x^M,\chi_M) \mid A_T] = 0$ for all $j\in [\ell]$, as $h_j$ is affine.
	We conclude that $y^T$ is a feasible solution to $P(T)$, and therefore $f(y^T,\chi_T) \geq P(T)$. Again appealing to Jensen's inequality, recalling that $y^T = \E_{M\sim \calI_k}[x^M \mid A_T]$ and that $\E_{M\sim \calI_k}[\chi_M \mid A_T] = \chi_T$ by \Cref{exp-apps}, we find that
	\begin{align*}
		\E_{M\sim \calI_k}[f(x^M,\chi_M) \mid A_T] & \geq f(y^T,\chi_T) \geq P(T).
	\end{align*}
	The lemma follows by total expectation over $M$, relying on $\Pr[A_T] = 1/\binom{|S|}{k}$ for each $T\in \binom{S}{k}$.
	\begin{align*}
		\E_{M\sim \calI_k}[P(M)] & = \sum_{T\in \binom{S}{k}}\E_{M\sim \calI_k}[P(M) \mid A_T] \cdot \Pr[A_T] \\
		& \geq \sum_{T\in \binom{S}{k}} P(T)\cdot \Pr[A_T] 
		= \E_{T\sim \calU_k}[P(T)]. \qedhere
	\end{align*}
\end{proof}
\section{Deferred Proofs of \Cref{sec:general}}
\label{sec:app-proofs}
In this section we provide the proofs deferred from \Cref{sec:general}.

\subsection{Implicit Tree Embedding}
\label{sec:implicit-tree}

In \Cref{sec:trees}, we proved that algorithm \algoname is $O(1)$-competitive on tree metrics. 
Therefore, as noted in \Cref{sec:general}, using tree embeddings and applying algorithm \algoname 
to the points according to distances in the obtained tree embedding yields an $O(\log n)$-competitive algorithm for general metrics.
Here we present an upper bound on \algoname's expected per-arrival cost which implies the same competitive bound, by relying on an \emph{implicit} tree embedding.

\TreeEmbed*

\begin{proof}
	For our proof we rely on low-stretch tree embeddings \cite{fakcharoenphol2004tight}. Given an $n$-point metric with distances $d_{i,j}$, this embedding is a distribution $\calD$ over tree metrics $T$ over the same point set, with tree distances $d^T_{i,j}$ satisfying the following for any two points $i,j$ in the metric.
	\begin{align}
	& d_{i,j} \leq d^T_{i,j}. \label{FRT-lower} \\ 
	& \E_{T\sim\calD} [d^T_{i,j}] \leq O(\log n)\cdot d_{i,j}.\label{FRT-upper}
	\end{align}
	
	For such a tree metric $T$, let $M^T(S)$ denote $M(S)$ with the distances $d_{i,j}$ replaced by $d^T_{i,j}$. (As before, we also let this denote the optimum value of this program.)
	By \eqref{FRT-lower} we immediately have that 
	$M(S) \leq M^T(S)$ for any set $S$,
	as any solution $\vec{x}$ to $M^T(S)$ is feasible for $M(S)$ and has lower cost for this latter metric, 
	$\sum_{i,j} x_{i,j}\cdot d_{i,j} \leq \sum_{i,j} x_{i,j} \cdot d^T_{i,j}.$ Consequently, we have 
	\begin{equation}\label{M-versus-tree-M}
	M(S)\leq \E_{T\sim\calD}[M^T(S)].
	\end{equation}
	
	Next, we denote by $OPT^T$ the optimum cost of the min-cost perfect matching of the requests to servers for distances $d^T_{i,j}$. By \Cref{lem:MSk-OPT-bound} we have that for a tree metric $T$
	\begin{equation}\label{M^T-versus-OPT^T}
	\E_{S_k\sim \calU_k}[M^T(S_k)] \leq \frac{4\cdot \E[OPT^T]}{\sqrt{nk}}.
	\end{equation}
	
	Finally, for any realization of requests, the minimum-cost matching of requests to servers under $d_{i,j}$ has expected cost (over the choice of $T$) at most $O(\log n)$ times higher under $d^T_{i,j}$, by \eqref{FRT-upper}. Therefore, by a coupling argument we get the following bound on $\E_{T\sim\calD}\E[OPT^T]$ in terms of $\E[OPT]$.
	\begin{equation}\label{OPT^T-versus-OPT}
	\E_{T\sim\calD}[OPT^T] \leq O(\log n)\cdot \E[OPT].
	\end{equation}
	
	Combining Equations \eqref{M-versus-tree-M}, \eqref{M^T-versus-OPT^T} and \eqref{OPT^T-versus-OPT}, we obtain our desired bound.
	\begin{align*}
	\E_{S_k\sim \calU_k}[M(S_k)] & \leq \E_{T\sim\calD}\E_{S_k\sim \calU_k}[M^T(S_k)] \leq \frac{4\cdot \E_{T\sim\calD} \E[OPT^T]}{\sqrt{nk}} \leq \frac{O(\log n)\cdot \E[OPT]}{\sqrt{nk}}.\qedhere
	\end{align*}
\end{proof}

\subsection{Load of $k$ Most Loaded Bins}
Here we prove our lower bounds on the sum of loads of the $k$ most loaded bins in a balls and bins process with $n$ balls and bins.

\topkballsandbins*
\begin{proof}
	Let $t = \frac{\log(n/k)}{\log\log\,(n/k)}$, and define
	\begin{gather*}
	f(x_1, \cdots, x_n) = \left\{
	\begin{matrix}
	1 & \text{if the } k^{th} \text{  largest number in } x_1, \cdots, x_n  \text{ is less than } t/2\\
	0 & \text{otherwise}\\
	\end{matrix}
	\right..
	\end{gather*}
	Clearly, the function $f(x_1, \cdots, x_n)$ satisfies the condition
	in \Cref{lem:poisson_approx}, i.e., $f(x_1, \cdots, x_n)$ is
	nonnegative and $\E[f(X_1^{m}, \cdots, X_{n}^{m})]$ is monotonically
	decreasing with $m$. Since we have an equal number of balls and bins, we consider the case $m=n$. We
	abbreviate $X_i^{n}$ to $X_{i}$ and $Y_{i}^{n}$ to $Y_{i}$. Let $M_k$
	be the $k^{th}$ largest number among $Y_1, \cdots, Y_n$. Applying
	\Cref{lem:poisson_approx},
	\begin{gather*}
		\Pr\left[L_{k} < t/2\right] = \E\left[f(X_1, \cdots, X_{n})\right] 
		\leq 2\cdot \E\left[f(Y_1, \cdots, Y_n)\right] = 2\cdot \Pr\left[M_k < t/2\right].
	\end{gather*}
	Define the indicator variable $Z_i := \mathbf{1}_{(Y_i \geq t/2)}$, and
	observe that $\Pr[M_k < t/2] = \Pr[ \sum_i Z_i < k]$. 
	We bound the
	latter via a Chernoff bound, so we need a lower bound on $\E[\sum_i Z_i]$.
	\begin{align}
		\E[\sum_i Z_i] = n\cdot \Pr[ Y_i \geq t/2] \geq  n\cdot 
		\Pr[Y_i = t/2] \stackrel{(a)}{=} \frac{n}{e(t/2)!}
		\stackrel{(b)}{\geq}  \frac{4n}{t!} \stackrel{(c)}{\geq}
		4k. \label{bound-z-expection} 
	\end{align}
	The equality~(a) uses the definition of the Poisson distribution, the
	inequality~(b) uses that $t! \geq 4e(t/2)!$ for sufficiently large
	$t$. For inequality~(c), we know $t! \leq \nicefrac{\sqrt{t}}{e} \,(t/e)^t$ from
	Stirling's approximation, and so when $n/k$ is sufficiently large,
	plugging in $t= \frac{\log(n/k)}{\log\log\,(n/k)}$ gives
	\begin{align*}
		\log(t!) &\leq (t+1/2)\log t - t - 1 \leq t\log t\leq \log(n/k).
	\end{align*}
	Putting things together, and using a Chernoff bound, we get
	\begin{gather*}
		\Pr\left[L_{k} < t/2\right] \leq 2\cdot \Pr\left[M_k < t/2\right] = 2\cdot \Pr[\sum_{i}Z_i < k ] \leq 2 e^{-\frac{(3/4)^2.4k}{2}} \leq 2e^{-k}.
	\end{gather*}
	The lemma then follows directly, as
	\begin{align*}
		\E[N_k] &\geq \E\left[N_k \mid L_k \geq t/2 \right] \cdot \Pr\left[ L_k \geq t/2 \right] 
		\geq k \cdot (t/2) \cdot (1 - 2e^{-k})
		=  \Omega\Big(\frac{k\cdot \log (n/k)}{\log \log (n/k)}\Big). \qedhere
	\end{align*}
\end{proof}

The following simple lemma states that in the min cost perfect matching, we can always match requests and servers in the same location as much as possible. That is, $x_{ii} = \frac{1}{n}$ for every requested location $i$.

\topkchernoff*
\begin{proof}
	In expectation, there are $n\left(1 - 1/n\right)^{n} \sim n/e$ empty
	bins, thus on average  one would expect  $1/(1-1/e) > 1.5$ balls in each
	non-empty bin. To make this intuition formal, let
	$t = (1 - 1/e + 0.01)n$ and define
	\begin{gather*}
		f(x_1, \cdots, x_n) = \left\{
		\begin{matrix}
		1 & \text{if more than } t \text{ of } x_1, \cdots, x_n \text{ are greater than 0}\\
		0 & \text{otherwise}.\\
		\end{matrix}
		\right.
	\end{gather*}
	It is easy to verify that the function $f(x_1, \cdots, x_n)$ is non-negative and
	$\E[f(X_1^{m}, \cdots, X_n^{m})]$ is monotonically increasing in
	$m$. Define the variable $Z_i := \mathbf{1}_{(Y_i > 0)}$; then
	$Z_i \sim \text{Bernoulli}( 1 - 1/e)$. \Cref{lem:poisson_approx} and a
	Chernoff bound now give that for sufficiently large $n$,
	\begin{gather*}
		\E[f(X_1, \cdots, X_n)] \leq 2\cdot \E[f(Y_1, \cdots, Y_n)] = 2\cdot \Pr\left[\sum_i Z_i > tn\right] \leq 2e^{-\frac{0.01^2\cdot(1-1/e)n}{2}} < 0.01.
	\end{gather*}
	Hence
	\begin{gather*}
		\E[N_t] \geq \E\left[N_t\mid f(X_1, \cdots, X_n) = 0 \right] \cdot \Pr\left[f(X_1, \cdots, X_n) = 0\right]  \geq n\cdot (1 - 0.01) =0.99n.
	\end{gather*}
	Finally, for $k \leq n/2 (\leq t)$, we have that indeed
	$\frac{\E[N_k]}{k} \geq \frac{\E[N_{t}]}{t} \geq \frac{0.99n}{(1 - 1/e + 0.01)n} \geq \frac{3}{2}$.
\end{proof}

\subsection{Further Deferred Proofs}
\MatchToSelf*
\begin{proof}
	Fix an optimal solution $x^*$ of $\calI$ of maximum $\sum_i x^*_{ii}$ among optimal solutions of $\calI$. Suppose for contradiction that there exists some $i\in S_k$ such that $x^*_{ii}<\min\{\ell_i,r_i\}$. WLOG $\ell_i \leq r_i$ and so there 
	exists some locations $j,j'$ such that $x^*_{ij} >0$ and $x^*_{j'i} > 0$. Let $\epsilon = \min\{x^*_{ij},x^*_{j'i}\}$. 
	Consider the solution $\tilde{x}$ obtained from $x^*$ by increasing $x^*_{ii}$ and $x^*_{j'j}$ by $\epsilon$ and decreasing $x^*_{ij}$ and $x^*_{j'i}$ by $\epsilon$. This $\tilde{x}$ is a feasible solution to $\calI$ (as sums of the form $\sum_i x_{ij}$ and $\sum_j x_{ij}$ are unchanged and $\tilde{x}\geq 0$). Moreover, we find that 
	\begin{align*}
	\sum_{ij} d_{ij}\cdot \tilde{x}_{ij} & = \left(\sum_{ij} d_{ij}\cdot x^*_{ij}\right) + \epsilon\cdot (d_{ii} + d_{jj'} - d_{ij} - d_{ij'}) \\
	& = OPT(\calI) + \epsilon\cdot (d_{jj'} - d_{ij} - d_{ij'}) \leq OPT(\calI),
	\end{align*}
	by triangle inequality. That is, $\tilde{x}$ is an optimal solution to $\calI$ with a higher $\sum_i x_{ii}$ than $x^*$, contradicting our assumption. The lemma follows. 
\end{proof}

\sumhalf*
\begin{proof}
	As noted in the proof of \Cref{balls-and-bins-bound-raw}, 
	by \Cref{lem:match-to-self}, the optimal value of $M(S_k)$ is equal to that of a min-cost bipartite perfect $b$-matching instance with left vertices associated with $S_k$ with demand $\frac{1}{k}-\frac{1}{n}$ and right vertices associated with $S\setminus S_k$ with demand $\frac{1}{n}$. Similarly, $M(S\setminus S_k)$ is equal to the same, but with each $i\in S_k$ having demand $\frac{1}{n}$ and each $i\in S\setminus S_k$ having demand $\frac{1}{n-k}-\frac{1}{n}$. That is, these programs are just scaled versions of each other, and we we have that for any $k \leq n/2$,
	\[M(S_k) = \frac{1/k - 1/n}{1/n} \cdot M(S\backslash S_k) =  \left(\frac{n}{k} - 1\right)\cdot M(S\backslash S_k) \geq M(S\backslash S_k).
	\]  
	Consequently, taking expectation over $S_k$ (equivalently, over $S\setminus S_k$), we find that  for any $k \leq n/2$, we have $E_{S_k\sim \calU_k}[M(S_k)] \geq E_{S_{n-k}\sim \calU_{n-k}}[M(S_{n - k})]$. The lemma follows.
\end{proof}
\section{Max Weight Perfect Matching Problem in \IID Model}
\label{sec:max-weight-matching}
Here we prove that, with a small modification, \algoname achieves the optimal competitive ratio, i.e $\nicefrac{1}{2}$, in the max weight perfect matching problem introduced in ~\cite{chang2018dispatch}. 
Here, rather than compute a minimum cost perfect matching, we are tasked with computing a maximum weight perfect matching, which need not correspond to a metric.
Since we are now in a maximization problem and we are no longer in a metric space, we will not make the assumption that the distribution of all requests is uniform among all servers. Moreover, we make the following modification to our algorithm: in each round of \algoname, instead of finding a min cost perfect matching, we would find the max weight perfect matching. Correspondingly, we change the notation for $M(T)$: instead of being a min cost perfect b-matching induced by the set of free servers $T$ and requests $R$, now $M(T)$ refers to the \emph{max} weight perfect b-matching between the set of free servers $T$ and requests $R$. More formally, we have
\begin{align}
	M(T) := \max & \sum_{i\in T,j\in R}  w_{i, j} \cdot x_{i,j}\label{eq:max-weight-lp} \\
	\text{s.t. }  \sum_{j \in T} x_{i,j} & = \frac{1}{|T|} \qquad \forall i\in T\notag \\
	\sum_{i \in R} x_{i,j} & = p_i  \qquad\,\,\,\, \forall j \in R \notag\\
	x &\geq 0.\notag
\end{align}
Generalizing \algoname, if $S_k$ is the realized set of free servers and $x^{S_k}$ an optimal solution to $M(S_k)$, 
then upon arrival of a request at location $i$ (which happens with probability $p_i$),
we randomly pick a server $s$ to match this request to, chosen with probability $x^{S_k}_{i,s}/p_i$.

\textbf{Difference compared to \cite{chang2018dispatch}.} We note that Chang et al.~\cite{chang2018dispatch} used a similar LP 
to $M(T)$. 
Essentially, they used $M(S)$, the program obtained by considering \emph{all} servers (and not just free ones). 
Following \cite{feldman2009online,haeupler2011online}, they refer to this as the optimum of the ``expected graph''.
Their algorithm picks a preferred server among all servers with probability $x^{S_k}_{r,s}/p_i$. If this server is already 
matched, in order to output a perfect matching they randomly (i.e., uniformly) pick an alternative server to match to.
Our algorithm does not need to fall back on a second random choice, as it only picks a server among free servers. 
As we shall see, our algorithm's analysis follows rather directly from
our analysis of \algoname for the minimization
variant.

A key observation is that the structure lemma (\Cref{structure}) still holds for our maximization variant of \algoname. We restate it here.

\begin{cla}\label{clm:structure-restated}(Structure Lemma, Restated)
	For each time $k$, the set $S_k$ is a uniformly-drawn $k$-subset of $S$; i.e., $S_k\sim \calU_k$. Consequently, the weight of the algorithm's output matching is $$\E[ALG] = \sum_{k = 1}^{n}\E_{S_k \sim \;\calU_k}[M(S_k)].$$
\end{cla}
\Cref{clm:structure-restated} holds due to the same argument in \Cref{structure}. Notice that all we needed in the proof of \Cref{structure} is that upon arrival of a request $r_k = i$ when there are $k$ free servers $S_k$ we match $r_k=i$ to a any free server $s$ with probability $x^{S_k}_{i,s}/p_i$, and so we use edge $(i,s)$ with probability precisely $x^{S_k}_{i,s}$. This implies that each free server $s\in S_k$ is matched with probability precisely $\frac{1}{k}$ and that the expected weight of the edge matched is precisely $\sum_{i\in S,j\in S_k} w_{i,j}\cdot x^{S_k}_{i,j}$.

Next, we note that $\E[OPT]$ can be upper bounded in terms of $M(S)$.
\begin{cla}\label{clm:opt-bound}
	$\E[\OPT] \leq n \cdot M(S).$
\end{cla}
The proof is exactly the same as \Cref{optd-bound}. See also \cite[Lemma 1]{chang2018dispatch}.

Now we can prove that the maximization variant of \algoname is $\nicefrac{1}{2}$ competitive for the max weight perfect matching problem in the \IID model.
\begin{thm}\label{thm:max-weight-matching}
	The max-weight variant of \algoname is $\nicefrac{1}{2}$ competitive. 
\end{thm}
\begin{proof}
	Letting $x^{S_k} \in \arg\max M(S_k)$ for every $S_k$, we have the following bound
	\begin{align*}
	\E_{S_k \sim \calU_k}[M(S_k)] &= \sum_{S_k }\frac{1}{\binom{n}{k}}\sum_{i \in S_k, j \in R}w_{i, j}\cdot x_{i, j}^{S_k} &\text{def. of }x^{S_k} \\
	&\geq \sum_{S_k}\frac{1}{\binom{n}{k}}\sum_{i \in S_k, j \in R} w_{i, j}\cdot x_{i, j}^{S} & \text{def. of }x^{S_k}\text{ and }M(S_k)   \\ 
	&= \sum_{i \in S}\Pr_{S_k \sim \calU_k} [i \in S_k] \cdot \sum_{j \in R}w_{i, j}\cdot x_{i, j}^{S}\\
	&= \frac{k}{n}\cdot M(S) &\text{def. of }M(S) \\
	&\geq \frac{k}{n^2}\cdot \E[\OPT]. &\Cref{clm:opt-bound}
	\end{align*}
	Summing these up, by the structure lemma (\Cref{clm:structure-restated}) we have
	\[
	\E[\ALG] = \sum_{i = 1}^{n}\E_{S_k \sim \calU_{k}}[M(S_k)] \geq \sum_{k = 1}^{n}\frac{k}{n^2}\cdot \E[\OPT] \geq \frac{1}{2}\cdot \E[\OPT].\qedhere 
	\]
\end{proof}
\section{Need for Metricity (or other assumptions)}
Here we outline simple examples showing that even under \IID arrivals, online minimum cost perfect matching does not admit even a polynomially-bounded competitive ratio. For unknown \IID, and therefore for random order and adversarial arrivals, even \emph{one edge} violating the triangle inequality is enough to rule out sub-exponential competitive ratio. For random order and adversarial arrivals one such edge is enough to cause the competitive ratio to be unbounded.

\begin{lem}\label{kiid-imp}
	The competitive ratio of any online min cost perfect matching algorithm under known \IID arrivals is at least $\Omega(2^{n/2}/n^3)$. This is true even under a uniform distribution and if the costs of all but $2$ request types obey triangle inequality.
\end{lem}

\begin{proof}
	Let $n$ be even and let $[n]$ be the set of servers. Consider the following set of request types (each with probability $1/n$ of being drawn at each arrival): the first $n-2$ request types have cost $1$ to be served by all servers. So far the instance corresponds to the uniform metric on $2n-2$ points. Now, second to last request type has cost $1$ to be served by serves in $[n/2]$, and cost $2^{n/2}$ to be served by serves in $[n/2+1,n]$, and the last request type has the exact opposite costs. When fewer than $n/2$ of the last two types arrive, $OPT$ is exactly $n$, whereas in the opposite case, which happens with probability at most $2^{-n/2}$ by standard Chernoff Bounds, $OPT$ is at most $n\cdot \exp^{n/4}$, and so $E[OPT]\leq 2n$. 
	On the other hand, with probability $\Omega(1/n^2)$, exactly one request from the last two request types arrives, and this is the last of all arrivals. In this case, as the algorithm must match $n-1$ servers before this arrival, with constant probability the sole remaining unmatched server has cost $2^{n/2}$ to match to this last request. Therefore, we have $E[ALG] = \Omega(2^{n/2}/n^2)$.
\end{proof}

A similar argument implies that for the unknown \IID arrival model, even a single edge which violates triangle inequality is enough to rule out sub-exponential competitive ratio .

\begin{lem}\label{uiid-imp}
	The competitive ratio of any online min-cost perfect matching algorithm under unknown \IID arrivals is at least $n^{n-2}/2$. This is true even under a uniform distribution and if the costs of all edges but one satisfy the triangle inequality.
\end{lem}
\begin{proof}[Proof (Sketch)]
	The distribution is similar to that of \Cref{kiid-imp}. We have $[n]$ denote the servers and have $n-1$ request types with service cost $1$ for each server. The final type has service cost $1$ for all servers except for one (unknown) server for which the service cost is $n^{n}$. Each request is drawn uniformly from this distribution. Unless $n$ copies of the last request type arrive (an even which happens with probability $1/n^n$), the cost of the optimal matching is $OPT=n$, and otherwise it is $n-1+n^n$, and so $E[OPT]\leq 2n$. On the other hand, with probability $\Omega(1/n)$, the special request type has exactly one arrival, and this is at the last time step, and so with probability $1/n$ this request's ``costly'' serve is the sole unmatched server, implying $E[ALG] = \Omega(n^{n-2})$.
\end{proof}

Finally, the same argument can show that the same input as in \Cref{uiid-imp}, with the sole costly edge being arbitrarily high, rules out any bounded competitive ratio, as having exactly one request of each type yields and input with $OPT=n$ but with $ALG$'s matching cost being unboundedly bad with probability $\Omega(1/n)$.

\begin{cor}
	The competitive ratio of any online min-cost perfect matching algorithm under random arrival order is \emph{unbounded}. This is true even if the costs of all edges but one edge satisfy the triangle inequality.
\end{cor}

\bibliographystyle{acmsmall}
\bibliography{abb,ultimate}

\end{document}